\RequirePackage{fix-cm}
\documentclass[smallextended]{svjour3}       
\smartqed  
\usepackage{graphicx}

\usepackage{xcolor,comment}

\begin{document}

\title{%
A generalized Vitali set from nonextensive statistics}


\author{Ignacio S. Gomez         
        }


\institute{Ignacio S. Gomez \at
              Instituto de F\'{i}sica, Universidade Federal da Bahia,
	Rua Barao de Jeremoabo, 40170-115 Salvador--BA, Brasil \\
              \email{nachosky@fisica.unlp.edu.ar}           
}

\date{\today}

\maketitle






\begin{abstract}
We address a generalization of the Vitali set
through a deformed translational property that stems from
a generalized algebra derived from the nonextensive statistics.
The generalization is based on the so-called $q$-addition
$x\oplus_{q} y=x+y+(1-q)xy$ for rational values of $q$,
where the ordinary formalism is recovered
when the control parameter $q \to 1$.
The generalized
Vitali set is non-measurable
for all rational parameter $\frac{1}{2}<q\leq1$,
but in the limit $q\rightarrow\frac{1}{2}$
the non-measurability cannot be guaranteed.
%
Furthermore, assuming measurability when $q\rightarrow\frac{1}{2}$,
then this must be positive.
%
%
Monotonicity, $\sigma$-additivity, $\sigma$-finiteness, and translational invariance
are generalized according to the structure of the $q$-addition
and of the $q$-integral.
\end{abstract}




\section{\label{sec:1} Introduction}
The Vitali set represents one of the most relevant objects in mathematics.
It concerns the measurability problem and uses the Axiom of Choice (AC)
for its construction.
%
It constitutes the first elementary example of
non-measurable set (in the sense of the Lebesgue measure) of real numbers,
found originally by Giuseppe Vitali \cite{Vit05}.
Indeed, there is no an unique Vitali set but rather an uncountable family of them
depending on the choice function given by the AC.
The Vitali set gave a big boost in the foundations of maths and logic,
for instance, by inspiring
implicitly Zermelo-Fraenkel models \cite{Fra73} that do not require the AC.
%
One of the intuitive requirements employed for the Vitali construction
is the invariance of the Lebesgue measure
under
translations,
thus presupposing a sum operation between the set to be translated
and the value of the translation.

On the other hand, accumulative
evidence in multiple phenomena like non-ergodicity,
long-range systems, etc. motivated
a generalization of the statistical mechanics
for addressing these issues in a appropriate way than the
traditional descriptions
\cite{Tsa88,tsallis-book}.
Parallel to this progress, several generalized algebras
were developed
from the mathematical background of the nonextensive
statistics \cite{Nivanen03,Bor04,Lob09,Tem11,SicTsa17}.
In these structures the usual operations of the real numbers
are generalized in terms a control parameter, usually
denoted by $q$ within this context.
%
%
%

The goal of this letter is to provide a generalization of the
Vitali set, by making use of the generalized
$q$-addition \cite{Bor04} and
the $q$-integral as a nonextensive version of the
Lebesgue integral as well as the AC,
that recover their usual definitions in the
limit $q\rightarrow1$ and whose
non-measurability cannot be demonstrated when $q\rightarrow\frac{1}{2}$.
%

The letter is organized as follows.
In Section \ref{sec:2} we briefly review
some generalized operations of the nonextensive algebra
that will be used.
%
The generalized Vitali set is built in
Section \ref{sec:3}
by employing the AC and the $q$-sum.
We prove its
monotonicity, $\sigma$-additivity, $\sigma$-finiteness and
$q$-translation invariance.
Non-measurability is shown for $\frac{1}{2}<q\leq 1$.
Particularly, non-measurability is not guaranteed
for $q\rightarrow\frac{1}{2}$ and
if the generalized Vitali set is assumed to be measurable
for $q\rightarrow\frac{1}{2}$, then its measure must be positive.
%
Here we also elaborates further remarks about the proposed
generalized Vitali set and the measure theory and the AC.
Finally, in Section \ref{sec:4}
some conclusions and perspectives are outlined.

\section{\label{sec:2}
Brief review on basic concepts}

We present the minimal notions and concepts to be used in this work.

\subsection{\label{sec:2-1}
Generalized addition, derivative and integral operations}
%
Based on a original motivation by the Tsallis entropy
$S_q$,
with $q \in \mathbf{R}$ the control parameter,
$k$ a dimensional positive constant
and $(p_1,\ldots,p_W)$ a discretized probability distribution,
\begin{eqnarray}\label{tsallis entropy}
 S_q[\{p_i\}]=k \frac{\sum_{i=1}^{W}p_i^q-1}{1-q},
\end{eqnarray}
generalizations of the logarithm and
the exponential, the so-called $q$-logarithm and
$q$--exponential are defined by \cite{Tsa94}
\begin{eqnarray}\label{q-functions}
 \begin{array}{l}
  \ln_{q}x=\frac{x^{1-q}-1}{1-q} \qquad (x>0),\\
  e_{q}(x)=[1+(1-q)x]_{+}^{\frac{1}{1-q}} \qquad (x\in\mathbf{R}),
 \end{array}
\end{eqnarray}
where $[A]_+ =\max\{A,0\}$.
From (\ref{q-functions}) it follows
%
\begin{eqnarray}\label{q-functions properties}
 \begin{array}{l}
   \ln_{q}(xy)=\ln_{q}x+\ln_{q}y+(1-q)\ln_{q}x\ln_{q}y, \\
   e_{q}(x)e_{q}(y)=e_{q}(x+y+(1-q)xy).
 \end{array}
\end{eqnarray}
%
From these relations, a generalization of the traditional
operations of the real numbers (called nonextensive algebra),
deformed by the
control parameter $q$,
is defined \cite{Nivanen03,Bor04}:
\begin{eqnarray}\label{q-algebra}
 \begin{array}{l}
  x\oplus_{q} y=x+y+(1-q)xy, \\
  x\ominus_{q} y=\frac{x-y}{1+(1-q)y} \qquad (y\neq \frac{1}{q-1}),\\
  x\otimes_{q} y=[x^{1-q}+y^{1-q}-1]_{+}^{\frac{1}{1-q}} \qquad(x,y>0), \\
  x\oslash_{q} y=[x^{1-q}-y^{1-q}+1]_{+}^{\frac{1}{1-q}} \qquad (x,y>0),
 \end{array}
\end{eqnarray}
that are the so-called $q$-sum, $q$-difference, $q$-multiplication, and $q$-division.

A deformed differential is defined through the $q$-difference \cite{Bor04}:
\begin{eqnarray}
 \begin{array}{lll}
  d_q x &:=& \displaystyle \lim_{x' \to x} x' \ominus_q x \\
        & =& \displaystyle \frac{dx}{1+(1-q)x}.
 \end{array}
\end{eqnarray}
The definition of a deformed number \cite{Bor98,Bor14}
\begin{eqnarray}
 \begin{array}{lll}
  x_q &:=& \ln (e_{q}(x)) \\
      & =& \displaystyle \frac{1}{1-q}\ln(1+(1-q)x)
 \end{array}
\end{eqnarray}
allows the identity
$ dx_q = d_q x $,
i.e., the differential of the deformed $q$-variable is equal to the
$q$-differential of an ordinary variable
(see \cite{Bor14}).
The $q$-derivative is defined as
\begin{eqnarray}
 \label{q-derivative}
 (D_q f)(x) = \frac{df}{d_q x}=(1+(1-q)x)\frac{df}{dx},
\end{eqnarray}
and, consistently, the $q$-integral
\begin{eqnarray}
 \label{q-integral}
 (I_q(f))(x) = \int f(x)d_q x.
\end{eqnarray}
Deformed $q$-numbers and deformed $q$-sum are related through
\begin{eqnarray}
 \label{q-number-q-sum}
 (x\oplus_q y)_q= x_q + y_q.
\end{eqnarray}

\subsection{The Vitali set
}\label{sec:2-2}
We review the standard construction of the Vitali set\cite{Whe77}.
We first recall the AC in its
traditional form.

\emph{Axiom of Choice (AC):
for every family $(B_i)_{i\in I}$ of nonempty sets
there exists a set composed by $(x_i)_{i\in I}$
with $x_i\in B_i$ for all I.}

\noindent
In this way, given a family of nonempty sets, the AC guarantees
the existence of a set whose elements belong to each member
of the family, by extracting exactly one element of each one of them.
The AC lies on the foundations of the mathematics
since it constitutes a tool for demonstrating
the existence of important notions,
like the existence of a basis for all vectorial space,
or the non-measurable Vitali set,
as described below.

For two arbitrary numbers $x,y$ in $[0,1]$
an specific relation in $[0,1]$, denoted by $x\sim y$,
can be defined by
\begin{eqnarray}\label{Vitali-relation}
x\sim y \Longleftrightarrow x-y\in \mathbf{Q}
\end{eqnarray}
being $\mathbf{Q}$ the set of rational numbers.
Since $\sim$ is an equivalence relation\footnote{$\sim$ is an equivalence relation if satisfies
reflexivity ($x\sim x$), symmetry ($x\sim y \rightarrow y\sim x$), and
transitivity $(x\sim y, y\sim z \rightarrow x\sim z)$.},
by applying the AC to the family of equivalence classes
a set $V$ is obtained, containing exactly
one representative of each equivalence class.
Let us assume $V$ measurable
with $\mu(A)=\int_{A}dx$, the Lebesgue measure (the usual integral)
on the measurable sets of $\mathbf{R}$.
Since $\mathbf{Q}$ is numerable
then in particular $\mathbf{Q}\cap [-1,1]$ is also numerable.
Thus, $\mathbf{Q}\cap [-1,1]$ can be enumerated by a sequence
$\{r_k\}_{k\in \mathbf{N}}$.
If $x\in[0,1]$, by the construction of $V$,
there exists $v\in V$ and $r_j\in[-1,1]$
with $x=v+r_j$. So, $[0,1]\subseteq \bigcup_k V+r_k$.
Also, if $z\in V+r_k \cap V+r_l$ we have
$z=v+r_k=v^{\prime}+r_l$ and
$v-v^{\prime}=r_k - r_l\in \mathbf{Q}$ so
$v=v^{\prime}$ and $r_k=r_l$. This shows that
the translated sets $V+r_k$ are pairwise disjoint.
Finally, $\bigcup_k V+r_k\in [-1,2]$ by construction.
Therefore,
\begin{eqnarray}\label{Vitali-inequality}
[0,1]\subseteq \bigcup_k V+r_k \subseteq [-1,2].
\end{eqnarray}
Since $\mu$ is monotonous
and satisfies $\sigma$-additivity,
Eq. (\ref{Vitali-inequality}) implies
%
\begin{eqnarray}\label{Vitali-inequality2}
1\leq \sum_{k=1}^{\infty} \mu(V+r_k) \leq 3,
\end{eqnarray}
where $\mu([0,1])=1$ and $\mu([-1,2])=3$. Now
by the translational invariance of $\mu$ is
$\mu(V+r_k)=\mu(V)$ for all $k$ and therefore
by inserting this in
(\ref{Vitali-inequality2})
\begin{eqnarray}\label{Vitali-inequality3}
1\leq \sum_{k=1}^{\infty} \mu(V) \leq 3. \nonumber
\end{eqnarray}
Due to these inequalities
it is clear that $\mu(V)$ cannot be infinite.
If $\mu(V)$ has a finite value the series
results infinite, which is a contradiction.
Hence, a value for $\mu(V)$ cannot be defined.
Thus, in this standard demonstration the elements used were: the AC,
the monotonicity, the $\sigma$-additivity and the translational invariance
of the Lebesgue measure.
Finally, this forces to admit the existence of non-measurable sets in $\mathbf{R}$.

\section{Generalizing the Vitali set}\label{sec:3}

The $q$-algebra and $q$-calculus are used in the present
generalization of the
Vitali set.
All the formalism recovers the standard one for $q\rightarrow1$.
Our strategy is to generalize the relation (\ref{Vitali-relation}):
\begin{eqnarray}\label{q-Vitali-relation}
x \sim_q y \Longleftrightarrow x=y\oplus_q r \ \ \ \textrm{with} \ \ \
r \in \mathbf{Q},
\ \ \
q \in \mathbf{Q}
\end{eqnarray}
The restriction $q \in \mathbf{Q}$ is
to guarantee that $\sim_q$ is an equivalence relation.
%

From (\ref{q-Vitali-relation})
$\sim_q$ results an equivalence relation:
Reflexivity:
$x\sim_q x$ since using (\ref{q-algebra}),
it results $x=x\oplus_q 0$.
Symmetry:
If $x\sim_q y$ it follows that
$x=y+r+(1-q)yr$ with $r\in \mathbf{Q}$.
So $y=\frac{x-r}{1+(1-q)r}=x \oplus_q \frac{-r}{1+(1-q)r}$
with $\frac{-r}{1+(1-q)r}\in\mathbf{Q}$, and $y\sim_q x$.
Transitivity:
If $x\sim_q y$ and $y\sim_q z$
we have $x=y\oplus_q r_1$ and $y=z\oplus_q r_2$.
By the associativity of the $q$-sum
$x=(z\oplus_q r_2)\oplus_q r_1=z\oplus_q (r_2\oplus_q r_1)$
with $r_2\oplus_q r_1 \in \mathbf{Q}$.
This implies $x\sim_q z$.

Now we can proceed with the Vitali construction in the usual form,
by applying the AC to the equivalence classes of $\sim_q$ in $[0,1]$,
denoted by $V_q$.
%
%
Some properties of $V_q$ can be summarized.
\begin{lemma}\label{lema-q-Vitali-properties}
Let $q$ be such that $0\leq q\leq 1$.
If $\{r_k\}$ is an enumeration of $\mathbf{Q}\cap [-1,1]$ the
following properties
are satisfied:
\begin{enumerate}
  \item[$(i)$]
   The $q$-translated sets $V_q \oplus_q r_k$ are disjoint pairwise.
  \item[$(ii)$]
   $[0,1]\subseteq \bigcup_k V_q \oplus_q r_k \subseteq [-2,3]$.
\end{enumerate}
\end{lemma}
\begin{proof}
$(i):$
If $z\in V_q\oplus_q r_k\cap V_q\oplus_q r_l$
we obtain $z=v\oplus_q r_k=v^{\prime}\oplus_q r_l$.
So
$v=v^{\prime}\oplus_q \left(\frac{r_l-r_k}{1+(1-q)r_k} \right)$
with
$\frac{r_l-r_k}{1+(1-q)r_k}\in \mathbf{Q}$, so $v=v^{\prime}$
and $r_k=r_l$.

\noindent
$(ii):$
If $x\in[0,1]$ there exists $v\in V_q\subset [0,1]$ and $r\in\mathbf{Q}$
with $x=v\oplus_q r$.
Since $0\leq x,v\leq1$ and $0\leq1-q\leq1$
this implies that
$-1\leq r=\frac{x-v}{1+(1-q)v}\leq 1$ with $r\in\mathbf{Q}$.
Thus,
$[0,1]\subseteq \bigcup_k V_q \oplus_q r_k$
with $r_k\in \mathbf{Q}\cap[-1,1]$.

Moreover, if $z\in V_q\oplus r_k$ then $z=v+(1+(1-q)v)r_k$
with $0\leq v\leq 1$, $-1\leq r_k\leq1$ and $0\leq1-q\leq 1$.
Joining these inequalities, it results
$-2\leq v(1+(1-q)r_k)+r_k\leq 3$.
Hence, $V_q \oplus_q r_k \subseteq [-2,3]$.
Thus, $\bigcup_k V_q \oplus_q r_k \subseteq [-2,3]$.
\end{proof}
Before
analyzing the measurability of $V_q$,
we need
to list
some properties of the $q$-integral.

\begin{lemma}
(Monotonicity, $\sigma$-additivity, $\sigma$-finiteness
and $q$-invariance translational of the $q$-integral)
\label{lema-q-lebesgue}

\noindent
Considering the measure $\mu_q(A)=\int_{A}d_q x$
on the $\sigma$-algebra
$\Sigma_q$ of subsets of real numbers
contained in $(\frac{-1}{1-q},+\infty)$
some properties are satisfied\footnote{Here we are considering $(\frac{-1}{1-q},+\infty)$
          is the universal set, so the complement of a subset A
          is understood to be $(\frac{-1}{1-q},+\infty)-A$.}:
\begin{enumerate}
  \item[$(a)$]
   $\mu_q$ is monotonous, $\sigma$-additive and $\sigma$-finite.
  \item[$(b)$]
   $\mu_q(A\oplus_q v)=\mu_q(A)$ \ $\forall A\in \Sigma_q$
   and $v\in (\frac{-1}{1-q},+\infty)$.
  \item[$(c)$]
  $\mu_q(\alpha A)=\alpha \mu_{1-(1-q)\alpha}(\alpha A)$ \ \ \
  $\forall\alpha>0$.
\end{enumerate}
\end{lemma}
\begin{proof}
$(a)$:
The monotonicity and $\sigma$-additivity is a consequence of
the properties of the Lebesgue integral and of
the definition of $\mu_q$ given by (\ref{q-integral}).
To see the $\sigma$-finiteness it is sufficient to notice
$
 (\frac{-1}{1-q},+\infty)
  = \left(\bigcup_{n=1}^{\infty} A_n\right)
    \cup
    \left(\bigcup_{k=1}^{\infty} B_k\right)
$
with
$
 A_n = [\frac{-1}{1-q}+\frac{1}{n+1},
        \frac{-1}{1-q}+\frac{1}{n})
$
for all $n\in\mathbf{N}
$
and
$B_k=[\frac{-1}{1-q}+k,\frac{-1}{1-q}+k+1)$ for all $k\in\mathbf{N}$
where
$\mu_q(A_i)=\mu_q(B_i)=\frac{1}{1-q}\ln(1+(1-q)\frac{1}{i})<\infty$
$\forall$
$i\in\mathbf{N}$.

\noindent
$(b)$:
Since $\Sigma_q$ is a $\sigma$-algebra and by the $\sigma$-additivity
it is sufficient to prove for intervals $[x_1,x_2]$.
Since $0\leq 1-q$ we have $[x_1,x_2]\oplus_q v=
[x_1\oplus_q v,x_2\oplus_q v]$
for $v\in (\frac{-1}{1-q},+\infty)$.
Then, by the definition of $\mu_q$ we obtain
$\mu_q([x_1,x_2]\oplus v)
  = \int_{x_1\oplus v}^{x_2\oplus v}d_qx
  = \int_{x_1\oplus v}^{x_2\oplus v}\frac{dx}{1+(1-q)x}
  = (x_2\oplus v)_q-(x_1\oplus v)_q
  = (x_2)_q+v_q-((x_1)_q+v_q)=(x_2)_q-(x_1)_q
  = \mu_q([x_1,x_2])$, where we have used that $(x\oplus y)_q=x_q+y_q$
  for all pair of real numbers $x,y$.

\noindent
$(c)$:
Let us show it for an arbitrary interval $[x_1,x_2]$.
Since $\alpha [x_1,x_2]=[\alpha x_1,\alpha x_2]$ then
$
 \mu_q(\alpha [x_1,x_2])
  = \int_{\alpha x_1}^{\alpha x_2}\frac{dx}{1+(1-q)x}
  = \frac{1}{1-q} \ln\left(
                            \frac{1+(1-q)\alpha x_2}{1+(1-q)\alpha x_1}
                      \right)
$
that can be written as
$
 \mu_q(\alpha [x_1,x_2])
  = \frac{1}{1-q^{\prime}}
    \ln \left(
              \frac{1+(1-q^{\prime}) x_2}{1+(1-q^{\prime}) x_1}
        \right)
$
with $q^{\prime}=1-(1-q)\alpha$.
This completes the proof.
\end{proof}
With the help of Lemmas \ref{lema-q-Vitali-properties}
and \ref{lema-q-lebesgue}
it is straightforward to study the measurability of $V_q$.
Thus, we arrive to our
main contribution of the present work.
\begin{theorem}
(Non-measurability of $V_q$ for $\frac{1}{2}<q\leq 1$)
\label{teorema-Vq}

\noindent
The generalized Vitali set $V_q$ results non-measurable
for $\frac{1}{2}<q\leq 1$.
However, when $q\rightarrow\frac{1}{2}$ the non-measurability
cannot be guaranteed.
Furthermore, if $V_{\frac{1}{2}}$ is measurable
then $\mu_{\frac{1}{2}}(V_{\frac{1}{2}})>0$.
\end{theorem}
\begin{proof}
From the monotonicity of $\mu_q$, and due to
Lemma \ref{lema-q-Vitali-properties} $(ii)$
\begin{eqnarray}
 \label{demo-1}
 \mu_q([0,1]) \leq \mu_q(\bigcup_k V_q \oplus_q r_k)
              \leq \mu_q([-2,3]),
 \nonumber
\end{eqnarray}
which can be recasted from the $\sigma$-additivity of $\mu_q$
(Lemma \ref{lema-q-lebesgue} $(b)$), as
\begin{eqnarray}
 \label{demo-2}
 \mu_q([0,1]) \leq \sum_{k=1}^{\infty}\mu_q(V_q \oplus_q r_k)
              \leq \mu_q([-2,3]).
 \nonumber
\end{eqnarray}
Now, since the sequence $\{r_k\}$ is in $[-1,1]$,
it must be $\{r_k\}\subseteq (\frac{-1}{1-q},+\infty)$
(because $\frac{-1}{1-q}<-1$ for $\frac{1}{2}< q\leq1$).
Hence, if we apply the $q$-invariance translational of $\mu_q$
(Lemma \ref{lema-q-lebesgue} $(b)$) to the last equation we obtain
\begin{eqnarray}
 \label{demo-qVitali}
 \frac{1}{1-q} \ln(2-q)
   \leq \mu_q(V_q) \left(
                         \sum_{k=1}^{\infty}1
                    \right)
    \leq \frac{1}{1-q} \ln \left(
                                 \frac{4-3q}{2q-1}
                           \right),
\end{eqnarray}
valid for all $\frac{1}{2}<q\leq 1$,
where $\mu_q([0,1])=\frac{1}{1-q}\ln(2-q)$ and
$\mu_q([-2,3]) = \frac{1}{1-q} \ln \left( \frac{4-3q}{2q-1} \right)$.

If $\mu_q(V_q)$ is finite (and non-zero) or infinite
from (\ref{demo-qVitali}) we have
$\frac{1}{1-q} \ln\left(\frac{4-3q}{2q-1}\right) = \infty$,
which is a contradiction since $\frac{1}{2}<q\leq 1$.
Neither $\mu_q(V_q)$ can be zero since $\frac{1}{1-q}\ln(2-q)>0$.
Thus, $V_q$ cannot be measurable for $\frac{1}{2}<q\leq 1$.

On the other hand, taking the limit $q\rightarrow\frac{1}{2}$
in (\ref{demo-qVitali}),
we have $\frac{1}{1-q}\ln(2-q)$ tends to $3\ln(\frac{4}{3})$ and
$\frac{1}{1-q}\ln\left(\frac{4-3q}{2q-1}\right)\rightarrow +\infty$,
so
\begin{eqnarray}
 \label{demo-qVitali2}
 3\ln\left(\frac{4}{3}\right)
   \leq \mu_{\frac{1}{2}} (V_{\frac{1}{2}}) \left(
                                            \sum_{k=1}^{\infty}1
                                            \right)
   \leq +\infty,
\end{eqnarray}
which does not lead to any contradiction.
Therefore, from (\ref{demo-qVitali2}) we cannot conclude the non-measurability
of $V_{\frac{1}{2}}$.
Finally, assuming
$V_{\frac{1}{2}}$ measurable and using (\ref{demo-qVitali2})
then $\mu_{\frac{1}{2}}(V_{\frac{1}{2}})$ must be positive.
This completes the proof.
\end{proof}

\subsection{A discussion concerning the measure theory}\label{sec:3-1}
By using the AC for demonstrating theorems,
some kind of counterintuitive situations
can appear, many of them,
concerning geometrical notions: non-measurable sets as the Vitali one,
Banach-Tarski theorem \cite{Tao11}, etc.
In a general way, these
results force to make one of the subsequent
concessions in order to have a fair definition of volume:
\begin{itemize}
  \item[$(A)$] The volume of a set changes when it is rotated or translated.
  \item[$(B)$] The volume of the union of two disjoint sets
               is not equal to the sum of their volumes.
  \item[$(C)$] There are non-measurable sets.
  \item[$(D)$] ZFC axioms (Zermelo-Fraenkel set theory with the AC)
               could be altered.
\end{itemize}
The measure theory chooses $(C)$, maintaining invariance
(against rotations and translations), additivity and
the ZFC axioms as logical and intuitive requisites
for a reasonable notion of volume.
Looking at the construction of the standard Vitali set $V$,
the AC is the fundamental element to conclude the non-measurability.

For the generalized Vitali set $V_q$ occurs the same within
the range of values $\frac{1}{2}<q\leq 1$,
where the translation and the integral are the
nonextensive
ones ($q$-translation and $q$-integral).
However, when $q=\frac{1}{2}$,
from the monotonicity, the $\sigma$-additivity,
the invariance
under
$\frac{1}{2}$-translations of $\mu_{\frac{1}{2}}$
along with the AC, the non-measurability of $V_{\frac{1}{2}}$
cannot be obtained, as it is
established by the Theorem \ref{teorema-Vq}.
Even more, if $V_{\frac{1}{2}}$ is measurable,
then it must have a positive measure.
Therefore, in the construction of the generalized Vitali set,
the crucial element to analyze the measurability
results to be the
non-additive index
$q$ rather than the AC,
contrarily to be expected.
Thus, the role played by the $q$-algebra
turns out relevant in the context of a measure theory where
the translation and the integral are replaced by their
nonextensive versions.
For $q=\frac{1}{2}$, the conditions $(A)$-$(D)$
can be adapted into the corresponding
ones, compatible with a measure theory provided with the
$q$-addition and the $q$-integral of the $q$-algebra:
\begin{itemize}
  \item[$(A^{\prime})$]
   The $q$-volume of a set is invariant against $q$-translations.
  \item[$(B^{\prime})$]
   The $q$-volume of the union of two disjoint sets
   is equal to the sum of their $q$-volumes.
  \item[$(C^{\prime})$]
   There are non-measurable sets for $\frac{1}{2}<q\leq1$,
   except probably for $q=\frac{1}{2}$.
  \item[$(D^{\prime})$]
   ZFC axioms could be not altered for $q=\frac{1}{2}$
   in a nonextensive measure theory%
   \footnote{Here we refer to the measure theory
             provided with the measure $\mu_q$.}.
\end{itemize}
We summarize our discussion in Table \ref{tab:1}.
\begin{table}
\caption{Some characteristics of the standard measure
and its nonextensive version
for the range of values $\frac{1}{2}\leq q\leq 1$. The novel fact is that,
for $q=\frac{1}{2}$
the non-measurability of $V_{\frac{1}{2}}$ is not followed,
even though the AC is used.
}
\label{tab:1}       
\begin{tabular}{lll}
\hline\noalign{\smallskip}
Properties & standard measure $\mu$ & nonextensive measure $\mu_q$  \\
\noalign{\smallskip}\hline\noalign{\smallskip}
monotonicity & YES & YES \\ \hline
$\sigma$-additivity & YES & YES \\ \hline
translational invariance & YES & NO ($q\neq1$) \\ \hline
$q$-translational invariance & NO ($q\neq1$) & YES \\ \hline
$\sigma$-finiteness & YES & YES, $\infty$ for $q=\frac{1}{2}$ \\ \hline
ZFC axioms & YES & YES \\ \hline
Vitali non-measurability & YES & YES ($q\neq\frac{1}{2}$) \\ \hline
\noalign{\smallskip}
\end{tabular}
\end{table}

\section{Conclusions}\label{sec:4}

We have presented an extension of the Vitali construction by making use of the
$q$-sum, instead of standard sum, in the context of generalized
algebraic operations on
the real numbers
inherited by nonextensive statistics,
and where the standard ones are recovered when $q\rightarrow1$ as a special case.
The traditional proof of the non-measurability remains valid only
within the range of the rational values $\frac{1}{2}<q\leq1$.
For $q=\frac{1}{2}$ we have showed that the AC is not sufficient to determine the 
non-measurability of the generalized Vitali set $V_{\frac{1}{2}}$.
%
%
In this manner,
a deformation of the algebraic structure of the operations may not lead
to the counterintuitive results 
of the measure theory
and geometry as the non-measurable sets or the Banach-Tarski theorem.

A reexamination of other geometrical constructions
(as the Vitali set) by employing deformed operations
(as the provided by the $q$-algebra)
are desirable to be explored in order to
study
the interplay between the axioms used to formalize geometrical
and physical notions.
This future proposal could play a complementary role with others
in the literature
as the Solovey's model \cite{Sol70},
where the existence of non-measurable sets for the Lebesgue measure
is not provable within ZF set theory without the AC.

\begin{acknowledgements}
This work was partially supported by the
National Institute of Science and Technology for Complex Systems (INCT-SC)
and CAPES.
\end{acknowledgements}



\end{document}